\documentclass[12pt, a4paper]{article}

\usepackage{titlesec}
\titleformat*{\subsection}{\normalsize\bfseries}
\usepackage{tocloft}
\usepackage{authblk}

\usepackage{amsmath,amssymb,amscd,amsthm,latexsym,accents,stmaryrd}
\usepackage{mathrsfs}
\usepackage{mathtools}
\usepackage{colonequals}
\usepackage{amsfonts}
\usepackage{bbm}
\usepackage[latin9]{inputenc}
\usepackage[all,cmtip]{xy}
\usepackage{enumerate}

\usepackage{url}
\usepackage{hyperref}
\usepackage{tikz}
\usetikzlibrary{cd}
\usetikzlibrary{decorations.text}
\usetikzlibrary{patterns}
\usetikzlibrary{arrows.meta}

\usetikzlibrary{automata,positioning}
\usepackage[parfill]{parskip}

\usepackage{svn-multi}
\svnidlong
{$LastChangedBy: kahl $}
{$LastChangedRevision: 88 $}
{$LastChangedDate: 2022-08-02 13:11:54 +0100 (Tue, 02 Aug 2022) $}
{$HeadURL: svn+ssh://kahl@servidor.math.uminho.pt/home/kahl/Svn/Articles/Concint/svn/concint.tex $}

\makeatletter
\def\thanks#1{\protected@xdef\@thanks{\@thanks
		\protect\footnotetext{#1}}}
\makeatother

\makeatletter
\newcommand{\subjclass}[2][2020]{%
	\let\@oldtitle\@title%
	\gdef\@title{\@oldtitle\footnotetext{#1 \emph{Mathematics subject classification:} #2}}%
}
\newcommand{\keywords}[1]{%
	\let\@@oldtitle\@title%
	\gdef\@title{\@@oldtitle\footnotetext{\emph{Key words and phrases:} #1.}}%
}
\makeatother

\makeatletter
\DeclareFontEncoding{LS2}{}{\@noaccents}
\makeatother
\DeclareFontSubstitution{LS2}{stix}{m}{n}
\DeclareSymbolFont{largesymbolsstix}{LS2}{stixex}{m}{n}
\DeclareMathDelimiter{\lbrbrak}{\mathopen}{largesymbolsstix}{"EE}{largesymbolsstix}{"14}
\DeclareMathDelimiter{\rbrbrak}{\mathclose}{largesymbolsstix}{"EF}{largesymbolsstix}{"15}

\newcommand{\nocontentsline}[3]{}
\newcommand{\tocless}[2]{\bgroup\let\addcontentsline=\nocontentsline#1{#2}\egroup}

\newcommand{\ubar}[1]{\mkern2mu\underline{\mkern-2mu #1\mkern-2mu}\mkern2mu}

\theoremstyle{definition}
\newtheorem{defin}{Definition}[section]
\newtheorem{rem}[defin]{Remark}

\newtheorem{exs}[defin]{Examples}
\theoremstyle{plain}
\newtheorem{theor}[defin]{Theorem}
\newtheorem{lem}[defin]{Lemma}
\newtheorem{prop}[defin]{Proposition}

\newtheoremstyle{dotless-thm}
{3pt}
{3pt}
{}
{}
{}
{.}
{.5em}
{}
\theoremstyle{dotless-thm}

\def\N{{\mathbb{N}}}

\def\im{{\mbox{im}}}

\def\A{{\mathcal{A}}}
\def\B{{\mathcal{B}}}

\def\T{{\mathcal{T}}}

\def\Q{{\mathcal{Q}}}
\def\R{{\mathcal{R}}}

\allowdisplaybreaks

\author{Thomas Kahl\thanks{This research was partially supported by FCT (\emph{Fundação para a Ciência e a Tecnologia}, Portugal) through projects UIDB/00013/2020 and UIDP/00013/2020.} }

\affil{\small{Centro de Matem\'atica,
	Universidade do Minho,\\ Campus de Gualtar,
	4710-057 Braga,
	Portugal\\
\texttt{kahl@math.uminho.pt}}
}

\begin{document}

\title{On the homology language of HDA models of transition systems}

\date{}




\maketitle 

\begin{abstract}
	
\noindent Given a transition system with an independence relation on the alphabet of labels, one can  associate with it a usually very large symmetric higher-dimensional automaton. The purpose of this paper is to show that by choosing an acyclic relation whose symmetric closure is the given independence relation, it is possible to construct a much smaller nonsymmetric HDA with the same homology language. 	
\end{abstract}


\section*{Introduction}

Higher-dimensional automata are a powerful  combinatorial-topological model for concurrent systems. A higher-dimensional automaton (HDA) is a precubical set (i.e., a cubical set without degeneracies) with an initial state, a set of final states, and a labeling on 1-cubes such that opposite edges of 2-cubes have the same label \cite{vanGlabbeek, Pratt}. 
An HDA is thus a (labeled) transition system (or an ordinary automaton) with a supplementary structure consisting of two- and higher-dimensional cubes. The transition system represents the states and transitions of a concurrent system. An \(n\)-cube in an HDA indicates that the \(n\) transitions starting at its origin are independent in the sense that they may occur in any order, or even simultaneously, without any observable difference. It has been shown in \cite{vanGlabbeek} that higher-dimensional automata are more expressive than the principal traditional models of concurrency. 

The fact that the two- and higher-dimensional cubes in an HDA represent independence of transitions suggests that the higher-dimensional topology of an HDA contains global information on independence of processes and components of the modeled concurrent system. In \cite{wehda}, the homology language has been devised as a homological tool to describe this global independence structure of an HDA. The homology language of an HDA is defined to be the image of the homomorphism induced in homology by a certain labeling chain map that leads from the cubical chain complex of the HDA to the exterior algebra on the alphabet of labels (see Section \ref{secHL}). The homology language can be computed and analyzed using software \cite{pg2hda}.

Transition systems are arguably the most fundamental model for concurrent systems. A natural way to turn a transition system into an HDA is to fill in empty squares and higher-dimensional cubes \cite{vanGlabbeek, GaucherCombinatorics, GoubaultMimram, transhda}. This approach requires some concept of independence in order to decide which cubes to fill in. Such a concept of independence may be given by an independence relation---i.e., an irreflexive and symmetric relation---on the alphabet of action labels. Independence relations play a fundamental role in trace theory in the sense of Mazurkiewicz \cite{Mazurkiewicz}. Asynchronous transition systems \cite{WinskelNielsen} are an important example of transition systems that come equipped with an independence relation on the set of labels. 

HDAs constructed from transition systems using a cube filling procedure based on a symmetric concept of independence, such as an independence relation on the set of labels, are usually very large because an empty \(n\mbox{-}\)cube representing the execution of \(n\) independent actions is filled in \(n!\) times. From a computational point of view, it is therefore desirable to have a way to produce smaller HDA models of transition systems.	It has been shown in \cite{transhda} that using a cube filling rule based on an asymmetric rather than an independence relation on the alphabet, one can construct an HDA model of a transition system where the independence of \(n\) actions in a state is represented by a single \(n\mbox{-}\)cube (at least if the transition system under consideration is deterministic). It is, of course, not to be expected that two HDAs constructed using different filling rules from the same  transition system will be equivalent in a meaningful sense. 

Every independence relation on an alphabet is the symmetric closure of an acyclic relation, i.e., a relation that is acyclic when seen as a graph (see Section \ref{secHM}). The purpose of this paper is to establish that the HDAs constructed from a transition system using cube filling rules based, respectively, on an independence relation on the alphabet and a generating acyclic relation have the same homology language. The two HDAs may therefore be regarded as equivalent from the point of view of their independence structures. As we will show in Section \ref{secMain}, the HDA constructed using the independence relation is the free symmetric HDA generated by the one constructed using the acyclic relation. The latter is thus a significantly smaller HDA model of the given transition system than the former.

\section{Precubical sets and HDAs} \label{SecPrel}
\begin{sloppypar}
This section presents fundamental material on precubical sets,  higher-dimensional automata, and their symmetric variants.
\end{sloppypar}


\subsection*{Precubical sets} \label{precubs}

A \emph{precubical set} is a graded set \(P = (P_n)_{n \geq 0}\) with \emph{face maps} \[d^k_i\colon P_n \to P_{n-1}\quad (n>0,\;k= 0,1,\; i = 1, \dots, n)\] satisfying the \emph{cubical identities} \[d^k_i d^l_{j}= d^l_{j-1} d^k_i \quad (k,l = 0,1,\; i<j).\] If \(x\in P_n\), we say that \(x\) is of  \emph{degree} or \emph{dimension} \(n\). The elements of degree \(n\) are called the \emph{\(n\)-cubes} of \(P\). The elements of degree \(0\) are also called the \emph{vertices} of \(P\), and the \(1\)-cubes are also called the \emph{edges} of \(P\). 
The \emph{\(i\)th starting edge} of a cube \(x\) of degree \(n > 0\) is the edge \[e_ix = d_1^{0} \cdots d_{i-1}^ {0}d_{i+1}^{0}\cdots d_n^{0}x.\]

A \emph{precubical subset} of a precubical set is a graded subset that is stable under the face maps. The \emph{\(n\)-skeleton} of a precubical set \(P\) is the precubical subset \(P_{\leq n}\) defined by \((P_{\leq n})_m = P_m\) for \(m \leq n\) and \(P_m = \emptyset\) else.

A \emph{morphism} of precubical sets is a morphism of graded sets that is compatible with the face maps. The category of precubical sets can be seen as the presheaf category of functors \(\square^{\textsf{op}} \to {\mathsf{Set}}\) where \(\square\) is the small subcategory of the category of topological spaces whose objects are the standard \(n\)-cubes \([0,1]^n\) \((n \geq 0)\) and whose nonidentity morphisms are composites of the \emph{coface maps} \(\delta^k_i\colon [0,1]^n\to [0,1]^{n+1}\) (\(k \in \{0,1\}\), \(n \geq 0\), \(i \in  \{1, \dots, n+1\}\)) given by \(\delta_i^k(u_1,\dots, u_n)= (u_1,\dots, u_{i-1},k,u_i \dots, u_n)\). 

The \emph{geometric realization} of a precubical set \(P\) is the quotient space \[|P|=\left(\coprod _{n \geq 0} P_n \times [0,1]^n\right)/\sim\] where the sets \(P_n\) are given the discrete topology and the equivalence relation is generated by
\[(d^k_ix,u) \sim (x,\delta_i^k(u)), \quad  x \in P_{n+1},\; u\in [0,1]^n,\; i \in  \{1, \dots, n+1\},\; k \in \{0,1\}.\] 
The geometric realization of a morphism of precubical sets \({f\colon P \to Q}\) is the continuous map \({|f|\colon |P| \to |Q|}\) given by \(|f|([x,u])= [f(x),u]\).

\subsection*{The precubical set of permutations}

The family of symmetric groups \(S = (S_n)_{n\geq 0}\) (with \(S_0 = \{id_\emptyset\}\)) is a precubical set with respect to the face maps given by 
\begin{align*}
	d^{k}_i\theta(j) &= \left \{ \begin{array}{ll}
		\theta(j), & j < \theta^{-1}(i),\; \theta(j) < i,\\
		\theta(j) - 1, & j < \theta^{-1}(i),\; \theta(j) > i,\\
		\theta(j+1), & j \geq \theta^{-1}(i),\; \theta(j+1) < i,\\
		\theta(j+1)-1,& j \geq \theta^{-1}(i),\; \theta(j+1) > i
	\end{array}\right.
\end{align*}
(see \cite{SHDA} and compare \cite{Krasauskas, FiedorowiczLoday}). Since, by definition, \(d^0_i\theta = d^1_i\theta\), we may simplify the notation by setting 
\[d^{\iffalse \bullet \fi}_i\theta = d^0_i\theta = d^1_i\theta.\]

Since \(S_1 = \{id\}\), the two elements of \(S_2\) have the same faces. In degrees \(\geq 3\), permutations are determined by their faces: 

\begin{prop}\label{ddeterminism}
	Let \(n \geq 3\), and let \(\sigma, \theta \in S_n\) such that \(d_i\sigma = d_i\theta\) for all \(i \in \{1, \dots, n\}\). Then \(\sigma = \theta\).	
\end{prop}

\begin{proof}
	It follows from \cite[Lemma 3.3]{SHDA} that it is enough to show that there exists an \(r\) such that \(\sigma(r) = \theta(r)\). Suppose that this is not the case. Set \(i = \sigma^{-1}(n)\) and \(j = \theta^{-1}(n)\). Then \(i, j < n\). Indeed, suppose that \(i = n\). Then for \(1 \leq r \leq n-1\), \(d_n\sigma(r) = \sigma(r)\) because \(r < n = \sigma^{-1}(n)\), \(\sigma (r) \leq n\), and \(\sigma(r) \not= \sigma(n) = n\). Since \(\theta(r) \not= \sigma(r)\) and \(\theta (r), \theta(r+1) \leq n\), we have \(r \geq \theta^{-1}(n)\), \(\theta(r+1) < n\), and \[\sigma(r) = d_n\sigma(r) = d_n \theta(r) = \theta(r+1).\]
	In particular, \(j = \theta^{-1}(n) \leq 1\). Hence \(j = 1\) and \(\theta(1) = n\). Set \(s = \sigma^{-1}(1)\). Since \(\sigma (s) \not= \sigma (n)\), \(s < n\). Since \(\theta(s+1) = \sigma(s) = 1\), \(\theta^{-1}(1) = s+1\). Since the values of \(\sigma\) are \(\geq 1\), we have
	\[
	d_1\sigma(1) = \left \{ \begin{array}{ll}
		\sigma(1) -1, & 1 < \sigma^{-1}(1) = s,\\
		\sigma(2) -1, & 1 = s.
	\end{array} \right.
	\] 
	Since \(n > 2\), we have \(\sigma (1), \sigma(2) \not= \sigma (n) = n\) and therefore \(d_1\sigma(1) \leq n-2\). On the other hand, \(1 < s+1 = \theta^{-1}(1)\) and \(\theta(1) = n > 1\) and therefore \[n -2 \geq  d_1\sigma(1) = d_1\theta(1) = \theta(1) -1 = n -1,\] which is impossible. It follows that \(i < n\). An analogous argument shows that \(j < n\). 
	
	Hence \( 1 \leq i, j \leq n-1\). We show that \(i < j\). Since \(i \geq i = \sigma^{-1}(n)\), \(\sigma(i+1) \leq n\), and \(\sigma(i+1) \not= \sigma(i) = n\), we have \(d_n\sigma(i) = \sigma(i+1)\). Hence \(d_n\theta(i) = \sigma(i+1) \not= \theta(i+1)\). Since \(\theta(i+1) \leq n\), it follows that \(i < j\). Analogously, \(j < i\). Since this is impossible, there must exist an \(r\) such that \(\sigma(r) = \theta(r)\). 
\end{proof}

For later use, we state the following fact from \cite{SHDA}:

\begin{prop} {\rm{ \cite[Prop. 3.6(i)]{SHDA}}} \label{prop36}
	Let \(P\) be a precubical set, and let \(n \geq 2\), \(1 \leq i < j \leq n\), \(k,l \in \{0,1\}\), \(x \in P_n\), and \(\theta \in S_n\). Then
	 \[d^k_{d^{\iffalse \bullet \fi}_{\theta(j)}\theta (i)}d^l_{\theta(j)}x = d^l_{d^{\iffalse \bullet \fi}_{\theta(i)}\theta(j-1)}d^k_{\theta(i)}x.\] 	
\end{prop}

\subsection*{Symmetric precubical sets}

A \emph{symmetric precubical set} is a precubical set \(P\) equipped with a \emph{crossed action} of \(S\) on \(P\), i.e.,  a morphism of graded sets \(S \times P \xrightarrow{} P\), \((\theta, x) \mapsto \theta\cdot x\) such that
\begin{itemize}
	\item for all \(n \geq 0\) and \(x \in P_n\), \(id\cdot x = x\);
	\item for all \(n \geq 0\), \(\sigma, \theta \in S_n\), and \(x \in P_n\), \((\sigma \cdot\theta)\cdot x = \sigma \cdot(\theta \cdot x)\);
	\item for all \(n \geq 1\), \(\theta \in S_n\), \(x \in P_n\), \(i \in \{1, \dots, n\}\), and \(k \in \{0,1\}\),  \[d^k_i(\theta \cdot x) = d^{\iffalse \bullet \fi}_i\theta\cdot d^k_{\theta^{-1}(i)}x.\]
\end{itemize} 
Symmetric precubical sets form a category, in which the morphisms are morphisms of precubical sets that are compatible with the crossed actions. We remark that the category of symmetric precubical sets is isomorphic to the presheaf category \(\mathsf{Set}^{\square_S^{\mathsf{op}}}\) where \(\square_S\) is the subcategory of the category of topological spaces whose objects are the standard \(n\)-cubes \([0,1]^n\) \((n \geq 0)\) and whose morphisms are composites of the coface maps \(\delta^k_i\) defined above and the \emph{permutation maps}  
\[t_\theta \colon [0,1]^n \to [0,1]^n, (u_1, \dots , u_n) \mapsto  (u_{\theta(1)} \dots, u_{\theta(n)}) \quad (n \geq 0,\, \theta \in S_n)\] 
(cf. \cite{GaucherCombinatorics, GoubaultMimram}). 

The \emph{free symmetric precubical set} generated by a precubical set \(P\) is the symmetric precubical set \(SP\) defined by 
\begin{itemize}
	\item \(SP_n = S_n \times P_n\) \((n \geq 0)\);
	\item \(d^k_i(\theta,x) = (d^{\iffalse \bullet \fi}_i\theta, d^k_{\theta^{-1}(i)}x)\) \(({n \geq 1}, {\theta \in S_n}, {x \in P_n}, {1 \leq i \leq n}, {k \in \{0,1\}})\); 
	\item \(\sigma \cdot (\theta, x) = (\sigma \cdot \theta, x)\) \((n \geq 0, \sigma, \theta \in S_n, x \in P_n)\).
\end{itemize}
The free symmetric precubical set is functorial, and the functor \(P \mapsto SP\) from the category of precubical sets to the category of symmetric precubical sets is left adjoint to the forgetful functor.

\subsection*{Higher-dimensional automata} \label{HDAdef}

Throughout this paper, we consider a fixed alphabet \(\Sigma\). 
A \emph{higher-di\-mensional automaton (HDA)} over \(\Sigma\) is a tuple \[\Q = (P, \imath , F, \lambda)\] where  \(P\) is a precubical set, \({\imath \in P_0}\) is a vertex, called the \emph{initial state}, \({F \subseteq P_0}\) is a (possibly empty) set of \emph{final states}, and \(\lambda \colon P_1 \to \Sigma\) is a map, called the \emph{labeling function}, such that  \(\lambda (d_i^0x) = \lambda (d_i^1x)\) for all \(x \in P_2\) and \(i \in \{1,2\}\) \cite{vanGlabbeek}. We say that an HDA \(\Q' = (P', \imath', F', \lambda')\) is a \emph{sub-HDA} of \(\Q\) and write \(\Q' \subseteq \Q\) if \(P'\) is a precubical subset of \(P\), \(\imath' = \imath\), \(F' = F\cap Q'_0\), and \(\lambda' = \lambda|_{Q'_1}\). The \emph{\(n\)-skeleton} of \(\Q\) is the sub-HDA \(\Q_{\leq n} = (P_{\leq n}, \imath, F, \lambda|_{(P_{\leq n})_1})\). Higher-dimensional automata form a category, in which a morphism from an HDA \(\Q = (P, \imath, F, \lambda)\) to an HDA \(\Q' = (P', \imath', F', \lambda')\) is a morphism of precubical sets  \(f\colon P \to P'\) such that \(f(\imath) = \imath'\), \(f(F) \subseteq F'\), and  \(\lambda'(f(x)) =  \lambda(x)\) for all \(x \in P_1\). 

A \emph{symmetric HDA} is an HDA \(\Q = (P, \imath, F, \lambda)\) equipped with a crossed action of \(S\) on \(P\). Symmetric HDAs form a category, in which the morphisms are morphisms of HDAs that also are  morphisms of symmetric precubical sets. The \emph{free symmetric HDA} generated by an HDA \(\Q = (P, \imath, F, \lambda)\) is the symmetric HDA \(S\Q = (SP, (id, \imath), S_0\times F, \mu)\) where  \(\mu(id,x) = \lambda (x)\) \((x \in P_1)\) and the crossed action is the one of \(SP\). The assignment \(\Q \mapsto S\Q\) defines a functor from the category of HDAs to the category of symmetric HDAs, which is left adjoint to the forgetful functor.

\section{The homology language and free symmetric HDAs} \label{secHL}

In this section, we define the homology language of an HDA and show that an HDA and the free symmetric HDA generated by it have the same homology language. We work over a fixed principal ideal domain, which we suppress from the notation.

\subsection*{Cubical chains and cubical homology} The \emph{cubical chain complex} of a precubical set \(P\) is the  nonnegative chain complex \(C_*(P)\) where \(C_n(P)\) is the free module generated by \(P_n\) and the boundary operator \(d\colon C_n(P) \to C_{n-1}(P)\) is given by \[dx = \sum \limits_{i=1}^{n}(-1)^i(d^0_ix -d^1_ix), \quad x \in P_n\;\; (n > 0).\]  
The \emph{cubical homology} of \(P\), denoted  by \(H_*(P)\), is the homology of \(C_*(P)\).

\subsection*{The homology language of an HDA}

Let \(\Q = (P, \imath, F, \lambda)\) be an HDA over \(\Sigma\). Consider the exterior algebra on the free module generated by \(\Sigma\), \(\Lambda (\Sigma)\). Recall that this is the quotient of the tensor algebra on the free module on \(\Sigma\) by the two-sided ideal generated by all elements of the form \(x\otimes x\) where \(x\) runs through the free module on \(\Sigma\) (see \cite{BourbakiAlgI} for more details). The exterior algebra \(\Lambda (\Sigma)\) is canonically graded by the exterior powers of the free module generated by \(\Sigma\). We view the graded module \(\Lambda(\Sigma)\) as a chain complex with \(d=0\) and define the \emph{labeling chain map} \[\mathfrak{l} \colon C_*(P) \to \Lambda (\Sigma)\] on basis elements \(x \in P_n\) by 
\[\mathfrak{l}(x) = \left \{\begin{array}{ll}
	1_{\Lambda(\Sigma)}, & n = 0, \vspace{0.2cm}\\
	\lambda(e_1x)\wedge \dots \wedge  \lambda(e_nx),&  n > 0.
\end{array} \right.\]
By \cite[Prop. 4.4.5]{labels}, the labeling chain map is indeed a chain map. We may therefore define the \emph{homology language} of \(\A\) to be the graded module
\[HL(\Q) = \im (\mathfrak{l}_*\colon H_*(P) \to H_*(\Lambda(\Sigma)) = \Lambda(\Sigma)).\]
Some fundamental properties of the homology language have been established in \cite{wehda}.

\subsection*{Simple cubical dimaps}

	A \emph{simple cubical dimap} from a precubical set \(P\) to a precubical set \(P'\) is a continuous map \({f\colon |P| \to |P'|}\) such that for all \(n \geq 0\) and \(x \in P_n\), there exist elements \(y \in P'_n\) and \(\theta \in S_n\) such that for all \(u \in [0,1]^n\),
	\[f([x,u]) = [y, t_\theta(u)].\]
	By \cite[Prop. 6.2.4]{labels}, \(y\) and \(\theta\) are uniquely determined by \(f\) and \(x\). We may therefore slightly abuse notation and write \(f(x)\) to denote \(y\). Obviously, the geometric realization of a morphism of precubical sets is a simple cubical dimap. In \cite{labels}, a more general concept of cubical dimap has been defined, hence the adjective \emph{simple}.
	
	A \emph{simple cubical dimap} from an HDA \(\Q = (P, \imath, F, \lambda)\) to an HDA \(\Q' = (P', \imath', F', \lambda')\) is a simple cubical dimap of precubical sets \(f\colon P \to P'\) that preserves the initial and the final states and that satisfies \(\lambda'(f(x)) =  \lambda(x)\) for all \(x \in P_1\). Our interest in simple cubical dimaps is motivated by the following fact: 
	
	\begin{prop} \label{HLinc} {\rm{\cite[Prop. 5.7.2]{wehda}}} 
		Let \(\Q\) and \(\Q'\) be HDAs such that there exists a simple cubical dimap \(\Q \to \Q'\). Then \(HL(\Q) \subseteq HL(\Q')\).	
	\end{prop}

\subsection*{The homology language of a free symmetric HDA}

Let \(\Q = (P,\imath,F,\lambda)\) be an HDA. Since there exists a morphism of HDAs \(\Q\to S\Q\), by Proposition \ref{HLinc}, \(HL(\Q) \subseteq HL(S\Q)\). We show in Theorem \ref{symcubical} below that actually equality holds. 

\begin{lem} \label{deltalem}
	Let \(\theta \in S_n\) \((n \geq 1)\), and let \(i \in \{1, \dots ,n\}\) and \(k \in \{0,1\}\). Then  
	\[\delta^k_{\theta^{-1}(i)}\circ t_{d_i\theta} = t_\theta \circ \delta^k_i \colon [0,1]^{n-1} \to [0,1]^n.\]
\end{lem}

\begin{proof}
	Let \((u_1, \dots, u_{n-1}) \in [0,1]^{n-1}\). We have \[\delta ^k_i (u_1, \dots, u_{n-1}) = (u_1, \dots, u_{i-1}, k, u_i \dots, u_{n-1}).\] For \(j \in \{1, \dots ,n\}\), set 
	\[v_j = \left \{ \begin{array}{ll}
		u_j, & 1 \leq j < i,\\
		k, & j = i,\\
		u_{j-1}, & i < j \leq n.
	\end{array}\right.\]
	Then we have
	\begin{align*}
		t_\theta\circ \delta ^k_i (u_1, \dots, u_{n-1}) &= t_\theta (v_1, \dots, v_n)\\
		&= (v_{\theta(1)}, \dots, v_{\theta(n)})\\
		&= (v_{\theta(1)} , \dots, v_{\theta(\theta^{-1}(i)-1)},k,v_{\theta(\theta^{-1}(i)+1)}, \dots, v_{\theta(n)})\\	
		&= \delta^k_{\theta^{-1}(i)}(v_{\theta(1)}, \dots, v_{\theta(\theta^{-1}(i)-1)},v_{\theta(\theta^{-1}(i)+1)}, \dots, v_{\theta(n)}).											
	\end{align*}
	We have 
	\[\delta^k_{\theta^{-1}(i)}\circ t_{d_i\theta}(u_1, \dots, u_{n-1}) = \delta^k_{\theta^{-1}(i)}(u_{d_i\theta(1)}, \dots, u_{d_i\theta(n-1)}).\]
	Since for \(j \in \{1, \dots, n-1\}\),
	\[u_{d_i\theta(j)} = \left \{ \begin{array}{ll}
		u_{\theta(j)}, & j < \theta^{-1}(i),\, \theta(j) < i,\\
		u_{\theta(j) - 1}, & j < \theta^{-1}(i),\, \theta(j) > i,\\
		u_{\theta(j+1)}, & j \geq \theta^{-1}(i),\, \theta(j+1) < i,\\
		u_{\theta(j+1)-1},& j \geq \theta^{-1}(i),\, \theta(j+1) > i
	\end{array}\right. = \left \{ \begin{array}{ll}
		v_{\theta(j)}, & j < \theta^{-1}(i),\\
		v_{\theta(j+1)}, & j \geq \theta^{-1}(i),
	\end{array}\right.\]
	the result follows.
\end{proof}

\begin{theor} \label{symcubical}
	 \(HL(\Q) = HL(S\Q)\).
\end{theor}

\begin{proof}
	We only have to show the inclusion \(HL(S\Q) \subseteq HL(\Q)\). By Proposition \ref{HLinc}, it suffices to construct a simple cubical dimap \(S\Q \to \Q\). Consider  the continuous map \(f\colon |SP| \to |P|\) defined by 
	\[f([(\theta,x),u]) = [x, t_{\theta} (u)]\quad ((\theta,x)\in (SP)_n, \, u \in [0,1]^n).\]
	This is well defined because, by Lemma \ref{deltalem}, for \((\theta,x)\in (SP)_n\), \(u \in [0,1]^{n-1}\), \(i \in \{1, \dots, n\}\), and \(k \in \{0,1\}\),
	\begin{align*}
	[x,t_{\theta}\circ \delta^k_i(u)] &= [x,\delta^k_{\theta^{-1}(i)}\circ t_{d_i\theta}(u)]
	= [d^k_{\theta^{-1}(i)}x,t_{d_i\theta}(u)].	
	\end{align*}
	By construction, \(f\) is a simple cubical dimap of precubical sets. We have \(f(\theta, x) = x\) for all \((\theta, x) \in SP\). Hence \(f\) preserves the initial and the final states. Moreover, for every edge \(x \in P_1\), \(\lambda (f(id,x)) = \lambda(x)\). It follows that \(f\) is a simple cubical dimap of HDAs.
\end{proof}


\section{HDA models of transition systems} \label{secHM}

Following \cite{transhda}, we define HDA models of transition systems. An HDA model can be constructed with respect to an arbitrary relation on the alphabet of labels. In this paper, we are interested in the case where this relation is an independence or an acyclic relation.

\subsection*{Transition systems and independence relations}
\label{transdef}

	A \emph{transition system} is a \(1\mbox{-}\)truncated extensional HDA, i.e., an HDA with no cubes of dimension \(\geq 2\) and no two edges with the same label and the same start and end vertices. An \emph{independence relation} is an irreflexive and symmetric relation on the alphabet of action labels \(\Sigma\). An independence relation equips the alphabet with a notion of concurrency: two actions are independent if they may be executed sequentially or simultaneously without any relevant difference. Independence relations play a fundamental role in trace theory \cite{Mazurkiewicz, Mazurkiewicz2}. The results of this paper apply, in particular, to asynchronous transition systems, which are transition systems over an alphabet with an independence relation satisfying certain conditions (see \cite{WinskelNielsen}).
	
\subsection*{Acyclic relations}

A relation \(\ltimes \) on a set \(X\) is called \emph{acyclic} if for all \(n \geq 1\) and \(x_1, \dots, x_n\in X\),
\[x_1 \ltimes x_2,\, x_2 \ltimes  x_3,\, \dots ,\, x_{n-1}\ltimes  x_n \Rightarrow {x_n \not\! \ltimes  x_1}.\]	
Note that an acyclic relation is irreflexive \((n = 1)\) and, moreover,  asymmetric \((n = 2)\). 

\begin{prop}
Let \(I\) be an independence relation on \(\Sigma\). Consider a totally ordered set \((S,\leq)\) and a map \(f\colon \Sigma \to S\) such that \[\forall\, a, b \in \Sigma : a \,I\, b \Rightarrow f(a) \not= f(b).\]
Then an acyclic relation \(\ltimes \subseteq \Sigma \times \Sigma\) whose symmetric closure is \(I\) is given by
\[a \ltimes b \ratio \Leftrightarrow a \,I\, b ,\; f(a) \leq  f(b).\]	
\end{prop}

\begin{proof}
	Let \(x_1, \dots, x_n\in \Sigma\) such that 
	\(x_1 \ltimes x_2,\, x_2 \ltimes  x_3,\, \dots ,\, x_{n-1}\ltimes  x_n\). Then 
	\[f(x_1) \leq  f(x_2) \leq \dots \leq f(x_n)\]
	and therefore \(f(x_1) \leq f(x_n)\). If 
	\(f(x_n) \not \leq f(x_1)\), then \({x_n \not\! \ltimes  x_1}\). If \(f(x_n) \leq f(x_1)\), then \(f(x_n) = f(x_1)\) and therefore \({x_n {\,\not\! I\,} x_1}\). Hence \(x_n \not\! \ltimes  x_1\) in this case too. Thus, \(\ltimes\) is acyclic.
	
	Let \(a,b \in \Sigma\). Since \(\leq\) is a total order, we have \(f(a) \leq f(b)\) or \(f(b) \leq f(a)\). Hence if \({a \,I\, b}\) \((\Leftrightarrow {b \,I\, a}) \), we have \({a \ltimes b}\) or \({b\ltimes a}\). If, conversely, \({a \ltimes b}\) or \({b \ltimes a}\), then \({a \,I\, b}\)  by definition of \(\ltimes\). Thus, \(I\) is the symmetric closure of \(\ltimes\). 
\end{proof}

\begin{exs}
	(i) Consider an independence relation \(I \subseteq \Sigma \times \Sigma\), and let \(\leq\) be a total order on \(\Sigma\). Then an acyclic relation \(\ltimes\) whose symmetric closure is \(I\) is given by 
	\[a \ltimes b \ratio \Leftrightarrow a \,I\, b, \; a \leq b.\]
	
	(ii) Let \(I\) be an independence relation on \(\Sigma\), and let \(\rm{pid} \colon \Sigma \to \N\) be a function associating with each label a process ID. If no two actions of the same process are independent, an acyclic relation \(\ltimes\) whose symmetric closure is \(I\) is given by 
	\[a \ltimes b \ratio \Leftrightarrow a \,I\, b, \; {\rm{pid}}(a) \leq {\rm{pid}}(b).\]
\end{exs}

\subsection*{HDA models} \label{HDAmodels}

Let \(\T = (X, \imath, F, \lambda)\) be a transition system, and let \(R\) be a relation on \(\Sigma\). The relation \(R\) does not have to satisfy any condition. We say that an HDA \(\Q = (Q, \jmath, G, \mu)\) is an \emph{HDA model of \(\T\) with respect to \(R\)} if the following conditions hold:
\begin{description}
	\item[HM1] \(\Q_{\leq 1} = \T\), i.e., \(Q_{\leq 1} = X\), \(\jmath = \imath\), \(G = F\), and \(\mu = \lambda\). 
	\item[HM2] For all \(x \in Q_2\), \(\lambda(d^0_2x)\;R\; \lambda(d^0_1x)\).		
	\item[HM3] For all \(m\geq 2\) and \(x,y \in Q_m\), if \(d^k_rx = d^k_ry\) for all \(r \in\{1,\dots ,m\}\) and \(k \in \{0,1\}\), then \(x = y\).		
	\item[HM4] \(\Q\) is maximal with respect to the properties  HM1-HM3, i.e., \(\Q\) is not a proper sub-HDA of any HDA satisfying HM1-HM3.
\end{description}
Condition HM1 says that \(\Q\) is built on top of \(\T\) by filling in empty cubes. By condition HM2, an empty square may only be filled in if the labels of its edges are related. Condition HM3 ensures that no empty cube is filled in twice in the same way. By condition HM4, all admissible empty cubes are filled in. 

It has been shown in \cite[Thm. 4.2, Cor. 4.5]{transhda} that an HDA model of a transition system with respect to a given  relation always exists and that its isomorphism class only depends on the isomorphism class of the transition system. 

For later use, we state the following two propositions from \cite{transhda}:

\begin{prop} {\rm{\cite[Prop. 4.3]{transhda}}} \label{prop43}
	 Let \(\T = (X, \imath, F, \lambda)\) be a transition system, and let $\Q = (Q, \imath, F, \lambda)$ be an HDA model of $\T$ with respect to a relation \(R\) on \(\Sigma\). Consider an integer $n \geq 2$ and $2n$ (not necessarily distinct) elements $x^k_i \in Q_{n-1}$ $(k \in \{0,1\}, i \in \{1, \dots ,n\})$ such that $d^k_ix^l_j = d^l_{j-1}x^k_i$ for all $1\leq i<j \leq n$ and $k,l \in\{0,1\}$. If $n=2$, suppose also that $\lambda(x^0_i) = \lambda(x^1_i)$ for $i \in \{1,2\}$ and that $\lambda(x^0_2) \,R\, \lambda(x^0_1)$. Then there exists a unique element $x \in Q_n$ such that $d^k_ix = x^k_i$ for all $i \in \{1, \dots ,n\}$ and $k \in \{0,1\}$.
\end{prop}

\begin{prop} {\rm{\cite[Prop. 4.7]{transhda}}} \label{prop47} 
	Let \(\T = (X, \imath, F, \lambda)\) be a transition system, and let \(\Q = (Q,\imath, F, \lambda)\) be an HDA satisfying HM1 and HM2 with respect to \(\T\) and a relation \(R\) on \(\Sigma\). Then 
	\(\lambda(e_ix) \,R\, \lambda (e_jx)\) 
	for all \(n \geq 2\), \(x \in Q_n\), and \(1 \leq i < j \leq n\).
	
\end{prop}

\section{The homology language of HDA models} \label{secMain}

Throughout this section, let \(\T = (X,\imath,F,\lambda)\) be a transition system, and let \(I\) be an independence relation on \(\Sigma\) that is the symmetric closure of an acyclic relation \(\ltimes\). Let furthermore  \(\A = (P, \imath,F, \lambda)\) and \(\Q = (Q,\imath,F,\lambda)\) be HDA models of \(\T\) with respect to \(I\) and \(\ltimes\), respectively. By the main result of this paper, \(\Q\) is a much smaller HDA model of \(\T\) than \(\A\) with the same homology language:

\begin{theor} \label{main}
	\(\A \cong S\Q\) and \(HL(\A) = HL(\Q)\).
\end{theor}

	By Theorem \ref{symcubical}, it is enough to prove that \(\A \cong S\Q\). This is an immediate consequence of the fact that \(S\Q\) is an HDA model of \(S\T\) with respect to \(I\), which will be established in  Proposition \ref{symHDAmodel} below after a number of preparatory results. The labeling function of \(S\T\) and \(S\Q\) will be denoted by \(\mu\). Recall that \(\mu\) is defined by \(\mu (id,x) = \lambda(x)\).

\begin{lem} \label{theta}
	Let \(\B = (B,(id,\imath),S_0\times F, \mu)\) be an HDA satisfying HM1 and HM2 with respect to \(S\T\) and \(I\). Then for each \(b \in B_n\) \((n \geq 2)\), there exists a unique permutation \(\vartheta_b\in S_n\) such that for all \(1 \leq i < j \leq n\),
	\[\mu(e_{\vartheta_b(i)}b) \ltimes \mu(e_{\vartheta_b(j)}b).\]
	 Moreover, setting \(\vartheta_b = id\) for \(b \in B_{\leq 1}\), the map \(\vartheta \colon B \to S\) is a morphism of precubical sets. 
\end{lem}

\begin{proof}
	Let \(b \in B_n\) \((n\geq 2)\). By Proposition \ref{prop47}, we have \(\mu(e_ib) \,I\, \mu(e_jb)\) for all \(1 \leq i < j \leq n\). Since \(I\) is irreflexive, it follows that the set \[M = \{\mu(e_ib)\,|\, i \in \{1, \dots, n\}\}\] has \(n\) elements. Moreover, since \(I\) is the symmetric closure of \(\ltimes\), we have \(\mu(e_ib) \ltimes \mu(e_jb)\) or \(\mu(e_jb) \ltimes \mu(e_ib)\) for all \(i \not=j\). Since \(\ltimes \) is acyclic, this implies that it is transitive on \(M\). Indeed, if \(\mu(e_ib) \ltimes \mu(e_jb)\) and \(\mu(e_jb) \ltimes \mu(e_kb)\), then \(k \not= i\) and \(\mu(e_kb) \not\! \ltimes \mu(e_ib)\), which implies \(\mu(e_ib) \ltimes \mu(e_kb)\). Hence \(\ltimes\) is a strict total order on \(M\).  
	Since \(M\) has \(n\) elements, it follows that there exists a unique permutation \(\vartheta_b \in S_n\) such that
	\[\mu(e_{\vartheta_b(i)}b) \ltimes \mu(e_{\vartheta_b(j)}b)\]
	for all \(1 \leq i < j \leq n\).
	
	It remains to check that \(\vartheta\) is a morphism of precubical sets. Let \(b \in B_n\) \((n\geq 1)\),  \(i \in \{1, \dots ,n\}\), and \(k \in \{0,1\}\). If \(n \leq 2\), then \(d_i\vartheta_b = id = \vartheta_{d^k_ib}\). Suppose that \(n \geq 3\). We have 
	\[e_jd^0_ib = \left \{ \begin{array}{ll}
	e_jb, & 1 \leq j < i,\\
	e_{j+1}b, & i \leq j < n.	
	\end{array}\right.\]
	Since parallel edges of a cube have the same label (see, e.g.,  \cite[Lemma 4.6]{transhda}), we have \(\mu(e_jd^1_ib) = \mu(e_jd^0_ib)\) for all \(j \in \{1, \dots , n-1\}\). Hence for \(k\in \{0,1\}\), 
	\[\mu(e_jd^k_ib) = \left \{ \begin{array}{ll}
	\mu(e_jb), & 1 \leq j < i,\\
	\mu(e_{j+1}b), & i \leq j < n.	
	\end{array}\right.\]
	We therefore have 
	\begin{align*}
		\mu(e_{d_i\vartheta_b(j)}d^k_ib) &= \left \{ \begin{array}{ll}
			\mu(e_{\vartheta_b(j)}d^k_ib), & j < \vartheta_b^{-1}(i),\; \vartheta_b(j) < i,\\
			\mu(e_{\vartheta_b(j) - 1}d^k_ib), & j < \vartheta_b^{-1}(i),\; \vartheta_b(j) > i,\\
			\mu(e_{\vartheta_b(j+1)}d^k_ib), & j \geq \vartheta_b^{-1}(i),\; \vartheta_b(j+1) < i,\\
			\mu(e_{\vartheta_b(j+1)-1}d^k_ib),& j \geq \vartheta_b^{-1}(i),\; \vartheta_b(j+1) > i		
		\end{array}\right.\\	
		&= \left \{ \begin{array}{ll}
			\mu(e_{\vartheta_b(j)}b), & j < \vartheta_b^{-1}(i),\\
			\mu(e_{\vartheta_b(j+1)}b), &  j \geq \vartheta_b^{-1}(i).	
		\end{array}\right.
	\end{align*}
	Let \(1 \leq j < r \leq n-1\). If \(j < r < \vartheta_b^{-1}(i)\), we have
	\[\mu(e_{d_i\vartheta_b(j)}d^k_ib) = 	\mu(e_{\vartheta_b(j)}b) \ltimes  	\mu(e_{\vartheta_b(r)}b) = \mu(e_{d_i\vartheta_b(r)}d^k_ib).\]
	If \(j <  \vartheta_b^{-1}(i) \leq r\), we have
	\[\mu(e_{d_i\vartheta_b(j)}d^k_ib) = 	\mu(e_{\vartheta_b(j)}b) \ltimes  	\mu(e_{\vartheta_b(r +1)}b) = \mu(e_{d_i\vartheta_b(r)}d^k_ib).\]
	If \( \vartheta_b^{-1}(i) \leq j < r\), we have
	\[\mu(e_{d_i\vartheta_b(j)}d^k_ib) = 	\mu(e_{\vartheta_b(j+1)}b) \ltimes 	\mu(e_{\vartheta_b(r +1)}b) = \mu(e_{d_i\vartheta_b(r)}d^k_ib).\]
	Thus, \(\vartheta_{d^k_ib} = d_i\vartheta_b\).
\end{proof}

\begin{lem} \label{phi}
	Let \(\B = (B,(id,\imath),S_0\times F, \mu)\) be an HDA satisfying HM1 and HM2 with respect to \(S\T\) and \(I\). Then there exists a unique morphism of graded sets \(\phi \colon B \to Q\) such that \(\phi_{\leq 1} \colon B_{\leq 1} \to Q_{\leq 1}\) is the isomorphism given by \((id,x) \mapsto x\) and such that \(d^k_i\phi (b) = \phi(d^k_{\vartheta_b(i)}b)\) for all \(b \in B_n\) \((n \geq 1)\), \(i \in \{1, \dots, n\}\), and \(k \in \{0,1\}\).
\end{lem}

\begin{proof}
	We construct \(\phi\) inductively. Let \(n \geq 2\), and suppose we have constructed \(\phi\) up to degree \(n-1\). Let \(b \in B_n\). Consider the elements \(y^k_i = \phi(d^k_{\vartheta_b(i)}b) \in Q_{n-1}\). Suppose that \(1 \leq i < j \leq n\). By the inductive hypothesis, Proposition \ref{prop36}, and  Lemma \ref{theta}, 
	\begin{align*}
	d^k_i y^l_j &= d^k_i\phi(d^l_{\vartheta_b(j)}b) 
	= \phi(d^k_{\vartheta_{d^l_{\vartheta_b(j)}b}(i)}d^l_{\vartheta_b(j)}b)
	= \phi(d^k_{d_{\vartheta_b(j)}\vartheta_{b}(i)}d^l_{\vartheta_b(j)}b)\\
	&= \phi(d^l_{d_{\vartheta_b(i)}\vartheta_b(j-1)}d^k_{\vartheta_{b}(i)}b)
	= \phi(d^l_{\vartheta_{d^k_{\vartheta_b(i)}b}(j-1)}d^k_{\vartheta_{b}(i)}b)
	= d^l_{j-1}\phi(d^k_{\vartheta_{b}(i)}b)\\
	&= d^l_{j-1}y^k_i.
	\end{align*}
	If \(n = 2\), we also have 
	\[\lambda(y^0_i) = \lambda(\phi(d^0_{\vartheta_b(i)}b)) = \mu(d^0_{\vartheta_b(i)}b) = \mu(d^1_{\vartheta_b(i)}b)  = \lambda(\phi(d^1_{\vartheta_b(i)}b)) = \lambda(y^1_i)\]
	for all \(i \in \{1,2\}\). Moreover, by Lemma \ref{theta},
	\begin{align*}
	\lambda(y^0_2) &= \lambda(\phi(d^0_{\vartheta_b(2)}b)) = \mu(d^0_{\vartheta_b(2)}b) = \mu(e_{\vartheta_b(1)}b)\\
	&\ltimes \mu(e_{\vartheta_b(2)}b) = \mu(d^0_{\vartheta_b(1)}b) = \lambda(\phi(d^0_{\vartheta_b(1)}b)) = \lambda(y^0_1).
	\end{align*} 
	By Proposition \ref{prop43}, there exists a unique cube \(y \in Q_n\) such that \(d^k_iy = y^k_i\) for all \(k \in \{0,1\}\) and \(i \in \{1, \dots, n\}\). We set \(\phi(b) = y\). This defines \(\phi\) in degree \(n\).
\end{proof}

\begin{lem} \label{HM12}
	\(S\Q\) satisfies HM1 and HM2 with respect to  \(S\T\) and \(I\).			
\end{lem}

\begin{proof}
	HM1: We have 
	\begin{align*}
		(S\Q)_{\leq 1} &= ((SQ)_{\leq 1}, (id,\imath),S_0\times F, \mu)
		= (SQ_{\leq 1}, (id,\imath),S_0\times F, \mu)\\
		&= (SX, (id,\imath),S_0\times F, \mu) 
		= S\T.
	\end{align*}

	HM2: Let \((\theta, x)\in ({SQ})_2 = S_2\times Q_2\). For \(i \in \{1,2\}\), \[\mu(d^0_i(\theta,x)) = \mu(d_i\theta,d^0_{\theta^{-1}(i)}x) = \mu(id,d^0_{\theta^{-1}(i)}x) = \lambda(d^0_{\theta^{-1}(i)}x). \]
	Thus if \(\theta = id\), 
	\[\mu(d^0_2(\theta,x)) = \lambda(d^0_{2}x) \ltimes \lambda(d^0_{1}x) = \mu(d^0_1(\theta,x)). \]
	If \(\theta\) is the transposition \((2 \;\, 1)\), 
	\[\mu(d^0_1(\theta,x)) = \lambda(d^0_{2}x) \ltimes \lambda(d^0_{1}x) = \mu(d^0_2(\theta,x)). \]
	In both cases, \(\mu(d^0_2(\theta,x)) \,I\, \mu(d^0_1(\theta,x))\).
\end{proof}

\begin{lem} \label{vartheta}
	 In \(S\Q\), \(\vartheta_{(\theta,x)} = \theta\) for all elements \((\theta,x )\). 		
\end{lem}

\begin{proof}
	Let \((\theta, x) \in (SQ)_n\). We may suppose that \(n \geq 2\). By \cite[Prop. 4.3]{SHDA}, \(e_i(\theta,x) = (id,e_{\theta^{-1}(i)}x)\). Hence 
	\[\mu(e_i(\theta,x)) = \mu(id,e_{\theta^{-1}(i)}x) = \lambda(e_{\theta^{-1}(i)}x).\]
	Thus, \(\mu(e_{\theta(i)}(\theta,x)) = \lambda(e_ix)\).   By Proposition \ref{prop47},
	\[\mu(e_{\theta(i)}(\theta,x)) = \lambda(e_ix) \ltimes \lambda(e_jx) = \mu(e_{\theta(j)}(\theta,x))\]	
	for all \(1 \leq i < j \leq n\). Consequently, by Lemma \ref{theta}, \(\vartheta_{(\theta,x)} = \theta\). 
\end{proof}

\begin{prop} \label{symHDAmodel}
	 \(S\Q\) is an HDA model of \(S\T\) with respect to \(I\).
\end{prop}

\begin{proof}
	By Lemma \ref{HM12}, we only have to show HM3 and HM4.
	
\begin{sloppypar}	
	HM3: Let \(m \geq 2\), and let \({(\sigma, x), (\theta, y) \in (SQ)_m = S_m \times Q_m}\) such that \(d^k_r(\sigma, x) = d^k_r(\theta, y)\) for all \(r \in \{1, \dots, m\}\) and \(k \in \{0,1\}\). Then \((d_r\sigma, d^k_{\sigma^{-1}(r)}x) = (d_r\theta, d^k_{\theta^{-1}(r)}y)\) for all \(r\) and \(k\). If \(m \geq 3\), this implies \((\sigma, x) = (\theta, y)\) by Proposition \ref{ddeterminism}. In the case \(m = 2\), it is enough to show that \(\sigma = \theta\). Suppose that this is not the case. Then we may assume that \(\sigma = id\) and \(\theta = (2 \;\,1)\). But then
	\[\lambda (d^0_1y) = \lambda (d^0_2x) \ltimes \lambda (d^0_1x) = \lambda (d^0_2y),\] 	
	which is impossible because \(\ltimes\) is asymmetric. 
\end{sloppypar}	

	HM4: Suppose that \(\B = (B, (id, \imath), S_0\times F, \mu)\) is an HDA satisfying conditions HM1--HM3 with respect to \(S\T\) and \(I\) that contains \(S\Q\) as a sub-HDA. We have to show that \(B = SQ\).  Since \(S\Q \subseteq \B\), the maps \(\vartheta \colon B \to S\) and \(\vartheta \colon SQ \to S\) coincide on \(SQ\). Let \(\phi \colon B \to Q\) be the map of graded sets of Lemma \ref{phi}. The corresponding map for \(S\Q\) is the map \(\psi \colon SQ \to Q\) given by \(\psi(\theta,x) = x\). Indeed, by Lemma \ref{vartheta}, \(d^k_i\psi(\theta,x) = d^k_i x = \psi(d_{\theta(i)}\theta,d^k_ix) = \psi(d_{\theta(i)}\theta,d^k_{\theta^{-1}(\theta(i))}x) = \psi(d^k_{\theta(i)}(\theta,x)) =
	\psi(d^k_{\vartheta_{(\theta,x)}(i)}(\theta,x))\). Since \(S\Q \subseteq \B\), the restriction of \(\phi\) to \(SQ\) is \(\psi\). 
	
	By HM1, \(B_{\leq 1} = (SQ)_{\leq 1}\). Let \(m \geq 2\), and suppose inductively that \(B_{<m} = (SQ)_{<m}\). Let \(b \in B_m\). By the inductive hypothesis, \(d^k_ib \in (SQ)_{m-1}\). Write \(d^k_ib = (\theta,x)\). We have 
	\begin{align*}
		d^k_i(\vartheta_b, \phi(b)) &= (d_i\vartheta_b, d^k_{\vartheta^{-1}_b(i)}\phi(b))
		= (\vartheta_{d^k_ib}, \phi(d^k_{i}b))
		= (\vartheta_{d^k_ib}, \psi(d^k_{i}b))\\
		&= (\vartheta_{(\theta,x)}, \psi(\theta,x))
		= (\theta,x)
		= d^k_ib.
	\end{align*}
	By HM3, it follows that \(b = (\vartheta_b, \phi(b)) \in (SQ)_m\).
\end{proof}

\providecommand{\bysame}{\leavevmode\hbox to3em{\hrulefill}\thinspace}
\providecommand{\MR}{\relax\ifhmode\unskip\space\fi MR }
\providecommand{\MRhref}[2]{%
  \href{http://www.ams.org/mathscinet-getitem?mr=#1}{#2}
}
\providecommand{\href}[2]{#2}

\end{document}